\documentclass[12pt]{revtex4}
\usepackage{amssymb,amsfonts,amsmath,amsthm,mathtext,braket}

\newtheorem{proposition}{Proposition}

\newcommand{\Tr}{\mathop{\mathrm{Tr}}}

\begin{document}

\title{On reconstruction of states from evolution induced by quantum dynamical semigroups perturbed by covariant measures}

\author{\firstname{Grigori}~\surname{Amosov}}
\email[E-mail: ]{gramos@mi-ras.ru}

\author{\firstname{Egor}~\surname{Baitenov}}

\author{\firstname{Alexander}~\surname{Pechen}}

\address{Steklov Mathematical Institute of Russian Academy of Sciences, Gubkina str. 8, Moscow 119991, Russia}

\received{} 

\begin{abstract} 
In this work, we show the ability to restore states of quantum systems from evolution induced by quantum dynamical semigroups perturbed by covariant measures. Our procedure describes reconstruction of quantum states transmitted via quantum channels and as a particular example can be applied to reconstruction of photonic  states transmitted via optical fibers. For this, the concept of perturbation by covariant operator-valued measure in a Banach space is introduced and integral representation of the perturbed semigroup is explicitly constructed. Various physically meaningful examples are provided. In particular, a model of the perturbed dynamics in the symmetric (boson) Fock space is developed as covariant measure for a semiflow of shifts and its perturbation in the symmetric Fock space, and its properties are investigated. Another example may correspond to the Koopman-von Neumann description of a classical oscillator with bounded phase space.
\end{abstract}

\keywords{} 

\maketitle

\section{Introduction}

Quantum dynamical semigroups describing open system dynamics play a crucial role to understand the structure of evolution of quantum systems over time. Bounded generators of dynamical semigroups have the Gorini-Kossakowski-Sudarshan-Lindblad (GKSL) form~\cite{GKS1976,Lindblad1976}. This form, known also as Linbladian or standard, can be extended to unbounded generators~\cite{Cheb, Hol}. Quantum master equations with unbounded dissipative generators were treated in works~\cite{Davies1976,Davies1977} and later studied in~\cite{Fannes1976,DVV1977,Pule1980} for bosonic systems described by completely positive maps on CCR-algebras. Generator of the master equation for a damped harmonic oscillator was derived based on the quantum Langevin equation~\cite{Accardi1990}. The analysis of uniqueness and trace-preserving property of the minimal quantum dynamical semigroups were studied~\cite{ChebotarevFagnola1993,AJP2006}. Criteria for the existence of stationary states for quantum dynamical semigroups with generally unbounded generators illustrated by quantum optics physical examples were found~\cite{FagnolaRebolledo2001}. Analysis of singular perturbations of positive and substochastic semigroups on the normal states  and abstract spaces of states was conducted~\cite{M-K-IDAQP2008,ALM-K-2011}. 
Singular perturbation of quantum stochastic differential equation for a quantum system and quantum oscillator described by unbounded operators and interacting with an environment was rigorously studied~\cite{GoughJSP2007}. Existence of uniquely determined minimal trace-preserving strongly continuous dynamical semigroups on the space of density matrices for a dynamical semigroup with unbounded repeated perturbation of an open system was proven~\cite{TZ-JMP2016}. Master equations with unbounded generators appear in the analysis of control of a quantum-mechanical oscillator~\cite{PBF2022}. Gaussian solutions for GKSL-type equations with multi-modal generators which are quadratic in bosonic or fermionic creation and annihilation operators are discussed~\cite{TeretenkovMN2016,TeretenkovIDAQP2019}.

There are examples of quantum dynamical semigroups having a physical meaning but going beyond the standard form~\cite{Hol2, Hol3, Werner}. In \cite{Hol4} it was suggested to consider a perturbation of quantum dynamical semigroup determined by the operator-valued measure on half-axis covariant with respect to the action of semigroup. Following this idea, we introduce such perturbations both at the level of Hilbert \cite {Amo, Amo2} and Banach \cite{Amo3} spaces. In this framework we suppose that all covariant measures can be divided into two classes.  The first class includes measures that are absolutely continuous with respect to the Lebesgue measure on the half-axis. At the same time, the cases when the density consists of bounded and unbounded operators are considered separately. The second class includes measures singular with respect to the Lebesgue measure. In this situation, the domain of the generator changes with perturbation. Note that our definition of a perturbation changing the domain of the generator was partly inspired by the work \cite{BY}. Note that momentum-like operators defined on half-axis do appear also in completely classical mechanics using Koopman-von Neumann formalism for classical mechanical particles with a bounded phase space~\cite{McCaulPRER2019}.

$C_0$-semigroups acting on the Banach space of nuclear operators $\mathfrak {T}(H)$ in a Hilbert space $H$ have the clearest physical meaning. With the requirement to preserve the positive cone, they determine the dynamics of the states of the quantum system in the Schr\"oedinger picture. In this case, conjugate semigroups acting on the algebra of all bounded operators $B(H)$ in $H$ determine the dynamics in the Heisenberg picture. Then covariant operator-valued measures should also be subject to the requirement of preserving the positive cone. The measure in this case defines a completely positive instrument for measuring a certain moment of achievement \cite {Hol4}.  Perturbation by such a measure sometimes allows to restore the conservativeness of the system \cite{Amo4}.

Along these lines, we introduce the general concept of perturbation by covariant operator-valued measure in a Banach space. We explicitly construct integral representation of the perturbed semigroup and study various physically meaningful examples. As a practical result, our procedure describes reconstruction of quantum states transmitted via quantum channels. As a particular example, it can be applied to reconstruction of photonic  states transmitted via optical fibers.  For this example, a model of the perturbed dynamics in the symmetric (boson) Fock space is explicitly constructed as covariant measure for a semiflow of shifts and its perturbation in the symmetric Fock space, and its properties are investigated. Another example of the considered dynamics may correspond to the Koopman-von Neumann description of a classical oscillator with bounded phase space.

In~\cite{Werner} the concept of no-event dynamical semigroup was introduced. Such a map transforms pure states into pure ones. We show that the dynamics described by the GKSL generator can be considered as a perturbation of a no-event semigroup. Moreover the class of such perturbations associated with operator-valued measures is wider than the GKSL dynamics. Every such perturbation gives rise to a new semigroup such that the initial and perturbed dynamics are connected by an integral equation over some operator measure. Thus, we can restore the perturbed state by means of the initial state presented by the no-event semigroup.

The structure of this paper is the following. In Sec.~\ref{Sec:2}, we introduce the general concept of perturbation by covariant operator-valued measure in a Banach space and provide various examples. In Sec.~\ref{Sec:3}, we show that GKSL generators can be considered as perturbations of so called no-event semigroups. In Sec.~\ref{Sec:4}, the bosonic case is investigated in details, for which a model of the perturbed dynamics in the symmetric (bosonic) Fock space is explicitly constructed as covariant measure for a semiflow of shifts and its perturbation in this Fock space. Conclusions Sec.~\ref{Sec:Conclusions} summarizes the results. 

Throughout the paper we shall use the general statements from the theory of $C_0$-semigroups \cite{EN, Y}.

\section{Covariant measures in a Banach space}\label{Sec:2}

Let $X$ be a Banach space and ${\mathcal T}_t:X\to X,\ t\in {\mathbb R}_+$ be a semigroup in $X$, i.e.
\begin{align*}
{\mathcal T}_{t+s}=&{\mathcal T}_t{\mathcal T}_s,\quad t,s\in {\mathbb R}_+,\\
{\mathcal T}_0=&{\rm I}.
\end{align*}
If orbits of ${\mathcal T}$ are continuous in $t$ with respect to some topology $\tau$, then ${\mathcal T}_t=\exp(t{\mathcal L})$, where ${\mathcal L}$ is a generator of ${\mathcal T}=\{{\mathcal T}_t,\ t\in {\mathbb R}_+\}$ with the domain ${\rm dom}(\mathcal {L})$ dense in $X$ in the same topology $\tau$. If a strong operator topology is taken as $\tau$, then ${\mathcal T}$ is said to be a $C_0$-semigroup. Also important is the case when $X=(X_*)^*$ for some Banach space $X_*$ and $\tau $ is $w^*$-topology.

{\bf Definition 1.} {\it Suppose that ${\mathfrak M}(B):{\rm dom}({\mathcal K})\to X$ is a finitely additive function defined on the $\sigma $-algebra $\mathfrak B$ of measurable subsets of ${\mathbb R}_+$ such that
$$
\mathfrak {M}(B_1\cup B_2)=\mathfrak {M}(B_1)+\mathfrak {M}(B_2),\ 
$$
for any disjoint $B_1,B_2\in \mathfrak {B}$.
Then, $\mathfrak {M}=\{\mathfrak {M}(B),\ B\in \mathfrak {B}\}$ is said to be an (unbounded) operator-valued measure on $\mathbb {R}_+$. If
\begin{equation}\label{covariant}
{\mathcal T}_t{\mathfrak M}(B)=\mathfrak {M}(B+t),\quad B\in {\mathfrak B},\quad t\in \mathbb {R}_+
\end{equation}
holds true, then the measure $\mathfrak M$ is said to be covariant with respect to $\mathcal T$.
}

{\bf Definition 2.} {\it The covariant measure $\mathfrak {M}$ is said to be absolutely continuous with respect to the Lebesgue measure on ${\mathbb R}_+$ if there exists the density $P(t):{\rm dom}(\mathcal {L})\to X$ such that
$$
{\mathfrak M}(dt)=P(t)dt.
$$ 
In the opposite case, $\mathfrak {M}$ is said to be singular with respect to the Lebesgue measure.}

{\bf Example 1.}  {\it Suppose that $M$ is a bounded operator in $X$. It follows that
$$
{\mathfrak M}([t,s))=\int \limits _t^s{\mathcal T}_rMdr,\quad 0\le t\le s,
$$
is a covariant measure possessing the bounded density
$$
P(t)={\mathcal T}_tM,\quad t\in {\mathbb R}_+.
$$
}

{\bf Example 2.} {\it Suppose that $M$ is a bounded operator in $X$. It follows that
$$
{\mathfrak M}([t,s))=({\mathcal T}_t-{\mathcal T}_s)M,\quad 0\le t\le s,
$$
is a covariant measure possessing the (unbounded in general) density
$$
P(t)={\mathcal T}_t{\mathcal L}M,\ t\in {\mathbb R}_+.
$$
}

{\bf Example 3.} {\it Suppose that $X=L^2({\mathbb R}_+)$ and $\mathcal {T}=S$ is the semigroup of right shifts acting on any $\eta \in L^2({\mathbb R}_+)$ by the formula
\begin{equation}\label{shift}
		(S_t\eta)(x)=
		\begin{cases}
			\eta(x-t),& x>t,\\
			0,& 0\leq x\leq t,
		\end{cases}
	\end{equation}
Denote $\chi _{[t,s)}\in L^2({\mathbb R}_+)$ the characteristic function of the interval $[t,s)$ and fix $e\in L^2({\mathbb R}_+)$. Then the formula
$$
{\mathfrak M}([t,s))\eta =\braket {e,\eta}\chi _{[t,s)},\ \eta \in L^2({\mathbb R}_+),\quad 0\le t\le s,
$$
determines a covariant measure for $S$ that is singular with respect to the Lebesgue measure.
}

Suppose that $X=(X_*)^*$ for some Banach space $X_*$ and $\mathcal T$ is a semigroup in $X$ with $w^*$-continuous orbits. Then one can define a preadjoint $C_0$-semigroup ${\mathcal T}_*$ acting on a Banach space $X_*$ by the formula
$$
x({\mathcal T}_{*t}(\omega ))={\mathcal T}_t(x)(\omega),\ x\in X,\quad \omega \in X_*,\quad t\in {\mathbb R}_+.
$$
Note that the continuities of orbits of a semigroup in the strong and weak topologies are equivalent \cite{Y}.
If $\mathfrak {M}$ is a covariant measure for $\mathcal T$, then there exists the preadjoint covariant measure ${\mathfrak M}_*$ for ${\mathcal T}_*$. The condition (\ref {covariant}) results in
$$
{\mathfrak M}_*(B){\mathcal T}_{*t}={\mathfrak M}_*(B+t),\quad B\in \mathfrak {B},\quad t\in {\mathbb R}_+.
$$

\section{Lindbladians as perturbations of no-event semigroups}\label{Sec:3}

Now let $X={\mathfrak T}(H)$ be the Banach space of all nuclear operators in a separable Hilbert space $H$. 
Suppose that $K$ is a maximal dissipative operator with the domain ${\rm dom}(K)$ which is dense in $H$, then $T_t=e^{tK}$ is a $C_0$-semigroup of contractions in $H$.
The $C_0$-semigroup ${\mathcal T}_{*t}:{\mathfrak T}(H)\to {\mathfrak T}(H)$ is said to be {\it no-event} if it is defined by the formula \cite{Werner}
$$
{\mathcal T}_{*t}(\omega )=T_t^*\omega T_t,\quad \omega \in {\mathfrak T}(H),\quad t\in {\mathbb R}_+.
$$
The meaning of the name is that such a semigroup translates rank-one operators into rank-one operators. The domain of the generator $\mathcal {L}_*$ of the semigroup ${\mathcal T}_{*t}=e^{t\mathcal {L}_*}$ is
$$
{\rm dom}({\mathcal L}_*)\ni \{\ket {\psi}\bra {\phi },\ \psi ,\phi \in {\rm dom}(K^*)\}.
$$ 

Suppose that $(L_j)$ are some operators in $H$ whose domain includes ${\rm dom}(K^*)$. Define map $\Lambda $ by the formula
\begin{equation}\label{Lambda}
\Lambda (\ket {\psi }\bra {\phi })=\sum \limits _j\ket {L_j\psi }\bra {L_j\phi },\ \psi ,\phi \in {\rm dom}(K^*).
\end{equation}
Suppose that the condition
$$
\sum \limits _j||L_j\psi ||^2\le -Re(\psi ,K\psi),\quad \psi \in {\rm dom}(K^*),
$$
is satisfied. Then, $\Lambda (\ket {\psi }\bra {\chi })\in {\rm dom}(\mathcal {L}_*)$ and the measure
$$
\mathfrak {M}_*([t,s))=\int \limits _t^s\Lambda {\mathcal T}_{*r}dr
$$
is correctly defined and satisfies the condition
$$
\mathfrak {M}_*([t,s)){\mathcal T}_{*r}=\mathfrak {M}_*([t+r,s+r)),\quad 0\le t<s,\quad r\in {\mathbb R}_+.
$$
The adjoint operator-valued measure $\mathfrak {M}$ acting on the algebra of all bounded operators $B(H)=({\mathcal T}(H))^*$ possesses the property (\ref {covariant}) with respect to the semigroup
$$
{\mathcal T}_t(x)=T_txT_t^*,\quad x\in B(H),\quad t\in {\mathbb R}_+.
$$

Let us define the Lindbladian $\breve {\mathcal L}$ by the formula
$$
{\breve {\mathcal L}}(x)=Kx+xK^*+2\sum \limits _jL_j^*xL_j.
$$
Then, there exists a minimal solution to the GKSL equation of the form \cite{Hol4}
\begin{equation}\label{Lind}
\frac {d}{dt}{\breve {\mathcal T}}_t(x)={\mathcal L}({\breve {\mathcal T}}_t(x)),\quad t\in {\mathbb R}_+.
\end{equation}
Equation (\ref {Lind}) can be represented in the form of integral equation
\begin{equation}\label{equ}
\breve {\mathcal T}_t-\int \limits _0^t\mathfrak {M}(ds)\breve {\mathcal T}_{t-s}=\mathcal {T}_{t},\quad t\in {\mathbb R}_+.
\end{equation}

The following two examples were discussed in~\cite{Bhat}.

{\bf Example 4.} {\it Let $H=L^2({\mathbb R}_+)$, $K=\frac {d^2}{dx^2},\ {\rm dom}(K)=\{\psi\ :\ \psi ''\in H,\ \psi (0)=0\}$
and $T_t=e^{tK}$ be a $C_0$-semigroup of contractions describing the diffusion with extinction at the point $x=0$ with the self-adjoint generator $K^*=K<0$. The semigroup $T=\{T_t,\ t\in {\mathbb R}_+\}$ has an integral representation in the form
$$
(T_t\eta)(x)=\frac {1}{\sqrt {4\pi t}}\int \limits _0^{+\infty }\left (\exp \left (-\frac {(x-y)^2}{4t}\right )-\exp \left (-\frac {(x+y)^2}{4t}\right )\right )\eta (y)dy,\ \eta \in H,\ t\in {\mathbb R}_+.
$$
Let $L=-\frac {d}{dx},\ {\rm dom}(L)=\{\psi \ :\ \psi '\in H,\ \psi (0)=0\}$, and define $\Lambda $ by (\ref {Lambda}) and the measure
\begin{equation}\label{mera2}
{\mathfrak M}_*([t,s))(\ket {\psi }\bra {\phi })=\int \limits _t^s\Lambda {\mathcal T}_{*r}dr.
\end{equation}
Then, the solution to (\ref {equ}) describes the quantum diffusion with extinction.
}

The following example shows that not all covariant measures associated with no-event semigroups are generated by operators $(L_j)$.

{\bf Example 5.} {\it Put $H=L^2({\mathbb R}_+)$, $K=-\frac {d}{dx},\ {\rm dom}(K)=\{\psi\ |\ \psi '\in H,\ \psi (0)=0\}$
and $T_t=e^{tK}$ is the semigroup of right shifts (\ref {shift}). Note that $K^*=\frac {d}{dx},\ {\rm dom}(K^*)=\{\psi\ |\ \psi '\in H\}$. Fix $\omega _0\in {\mathfrak T}(H)$ and put
$$
\Lambda (\ket {\psi }\bra {\phi })=\psi (0)\overline \phi (0)\omega _0,\ \psi ,\phi \in {\rm dom}(K^*).
$$
It defines the absolutely continuous measure in the following way
\begin{equation}\label{mera}
{\mathfrak M}_*([t,s))(\ket {\psi }\bra {\phi })=\int \limits _t^s\Lambda {\mathcal T}_{*r}(\ket {\psi }\bra {\phi })dr=
\int \limits _t^s\psi (r)\overline {\phi }(r)dr\cdot \omega _0,\quad \psi ,\phi \in {\rm dom}(K^*).
\end{equation}
It implies that for the semigroup ${\mathcal T}_{t}(\cdot )=T_t\cdot T_t^*$ acting on $B(H)$ we obtain
$$
{\mathcal T}_t{\mathfrak M}(B)={\mathfrak M}(B+t),\quad B\in \mathfrak {B},\quad t\in {\mathbb R}_+.
$$
} 
While the generator of evolution $K=-\frac {d}{dx}$ is less natural for quantum mechanics than $K=-\frac {d^2}{dx^2}$, such generator naturally appears as a Koopman operator in the Koopman-von Neumann description of a classical oscillator with bounded phase space~\cite{McCaulPRER2019}. In the Koopman-von Neumann formalism, evolution of the probability distribution $\rho$ of a classical mechanical system with phase space $\cal P$ is represented via evolution of a vector $\psi$ in the Hilbert space $H=L^2({\cal P},d\mu)$ as $\rho=|\psi|^2$, where vector $\psi$ satisfies the Schr\"odinger like equation $\dot\psi =i K\psi$ with Koopman operator $K$. For an oscillator (with unit frequency $\omega=1$) in the action-agnle representation $(J,\theta)$, the Koopman operator formally has the form $K=-i\frac{\partial}{\partial\theta}$. Its domain is determined by the conditions for the system on the boundary of the phase space~\cite{McCaulPRER2019}.

\section{Covariant measures in the symmetric Fock space}\label{Sec:4}

For a detailed discussion of the notion of a symmetric (bosonic) Fock space $\Gamma_s(\mathcal {H})$ over one-particle Hilbert space $\mathcal H$, we refer the reader to the well-known monograph \cite{Partha}. An important for practical applications example of a bosonic Fock space describes photons in quantum optics~\cite{ScullyZubairyBook}.  In this case, vectors of $\Gamma_s(\mathcal {H})$ can be interpreted as multiphoton states and in the particular case of $\mathcal {H}=L^2(\mathbb R_+)$ the dynamics of the system can be considered as transmission of light through optical fibers. As it was mention in Section III we can look at the GKSL equation like on a perturbation of the no-event semigroup by some measure. The GKSL equation can be embedded into the symmetric Fock space, where it becomes a quantum stochastic differential equation. In this section we give the example in which our measure is singular such that deriving the corresponding quantum stochastic differential equation seems to be impossible. We will not touch here on the problem of unbounded observables and measures, which is certainly of separate interest \cite{Amo5}.

Let $H=\Gamma_s(L^2(\mathbb R_+))$ and set $K$ to be the differential second quantisation of  the operator $-\frac d{dx}$ with domain $\{\psi\ :\ \psi '\in L^2(\mathbb R_+),\ \psi (0)=0\}$. Then 
the semigroup $
T_t = e^{tK}
$
acts on the exponential vectors as
\begin{equation}\label{eq:T*te(f)}
    T_{t} e(f) = e(S_tf),\ f\in L^2(\mathbb R_+), 
\end{equation}
provided that $S$ denotes the semigroup of right shifts (\ref{shift}). The semigroup $\mathcal T_{*t}$ describes the left shifts in the Fock space along semi-axis with extinction. Physically, the particles moving left through the point $x=0$ are assumed to be destructed that is projected onto vacuum state. 

Maps $\mathcal T_{*t}$ and $\mathcal T_t$ act on rank one operators made of exponential vectors as
\begin{align*}
    \mathcal T_{*t}\ket{e(f)}\bra{e(g)} =&T_t^*\ket{e(f)}\bra{e(g)}T_t= \ket{e(S_t^*f)}\bra{e(S_t^*g)},\\ 
    \mathcal T_{t}\ket{e(f)}\bra{e(g)} = &T_t\ket{e(f)}\bra{e(g)}T_t^*=\ket{e(S_tf)}\bra{e(S_tg)},
\end{align*}
$f,g\in f\in L^2(\mathbb R_+)$.
The adjoint semigroup $\mathcal T$ maps any operator $x$ from $B(H)$ onto tensor product of the projector $\ket\Omega\bra\Omega$ in the space $\Gamma_s(L^2(0,t))$ and the copy of $x$ in the space $\Gamma_s(L^2(t,+\infty)\cong H$, where $H$ is naturally identified with $\Gamma_s(L^2(0,t))\otimes\Gamma_s(L^2(t,+\infty))$.

Consider the other semigroup $\breve{\mathcal T}$ in $B(H)$ whose action on $x$ differs from $\mathcal T$ by the tensor multiplier in $\Gamma_s(L^2(0,t))$ that is $\rm I$ instead of $\ket\Omega\bra\Omega$. In particular,
\begin{equation*}
    \breve{\mathcal T}_{t} \ket{e(f)}\bra{e(g)} = {\rm I}\otimes \ket{e(S_tf)}\bra{e(S_tg)}, 
\end{equation*}
where the latter tensor factorisation is with respect to $\Gamma_s(L^2(0,t))\otimes\Gamma_s(L^2(t,+\infty) $.
\begin{proposition}
The preadjoint semigroup  $\breve{\mathcal T}_{*t}$
acts on exponential rank one operators as
\begin{equation*}
    \breve{\mathcal T}_{*t} \ket{e(f)}\bra{e(g)} =\exp\left(\int_0^t\overline{g(x)}f(x)\, dx\right) \ket{e(S_t^*f)}\bra{e(S_t^*g)},\ f,g\in L^2(\mathbb R_+)
\end{equation*}
\begin{proof}
We have for $f,g,h_1,h_2\in L^2(\mathbb R_+)$ that
    \begin{align*}
        \left\langle e(h_1),\left(\breve{\mathcal T}_{*t} \ket{e(f)}\bra{e(g)}\right)e(h_2)\right\rangle=&\left\langle e(g),\left(\breve{\mathcal T}_{t}\ket{e(h_2)}\bra{e(h_1)}\right) {e(f)}\right\rangle \\ 
        =&\Braket{e(g),\left(I\otimes \ket{e(S_th_2)}\bra{e(S_th_1)}\right) {e(f)}}\\
        =&\exp\left(\int_0^t\overline{g(x)}f(x)\, dx +\braket{g,S_th_2}+\braket{S_th_1,f} \right) \\
        =&\left\langle e(h_1),\left[ \exp\left(\int_0^t\overline{g(x)}f(x)\,dx\right)\ket{e(S_t^*f)}\bra{e(S_t^*g)} \right]e(h_2)\right\rangle ,       
    \end{align*}         
\end{proof}

\end{proposition}
Physically, preadjoint semigroup $\breve{\mathcal T}_{*t}$ corresponds to left shifts with forgetting, i.e. the particles moving left through the point $x=0$ are considered as passing from the system to the environment. To show this, note that $\breve{\mathcal T}_{*t}$ takes the partial trace of the density matrix over $(0,t)$.

Notice that $\mathcal T_{*t}$ is no-event semigroup while $ \breve{\mathcal T}_{*t} $ is conservative.

Now we construct an operator-valued measure $\mathfrak M$ on $\mathbb R_+$ covariant with respect to $\mathcal T$ that perturbs $\mathcal T$ to $\breve {\mathcal T}$ via equation~\eqref{equ}. 
The same problem in case of anti-symmetric Fock space was considered in~\cite{Amo4}. In fact, we will map $H$ onto another Hilbert space $H^\wedge$ via some isometric isomorphism $W$ and deal with semigroups in $H^\wedge$.

Denote $\chi ([t,s))$ the operator of multiplication on characteristic function of $[t,s)$ in $L^2(\mathbb R_+)$. Then $\chi(dt)$ is projector-valued measure. Define the operator $$W\colon H\to H^{\wedge}:= \{\mathbb C\Omega\}\oplus H\otimes L^2(\mathbb R_+) $$
initially by its action on the exponential vectors,
\begin{equation}\label{eq:Wdef}
    We(f) = \Omega\oplus\left(  \int_0^{+\infty}e(S_t^*f)\otimes \chi(dt) f\right), \quad  f \in L^2(\mathbb R_+).
\end{equation}
\begin{proposition}
    $W$ preserves inner product of exponential vectors. 
\end{proposition}
\begin{proof}
Indeed,
    \begin{align*}
        & \Braket{\Omega\oplus\left(  \int_0^{+\infty}e(S_t^*f)\otimes \chi(dt) f\right),\Omega\oplus\left(  \int_0^{+\infty}e(S_t^*g)\otimes \chi(dt) g\right)} \\
        &=1+\int_0^{+\infty} e^{\braket{S_t^*f, S_t^*g}}\overline{f(t)} g(t)\, dt\\
        &=1+\int_0^{+\infty} e^{\braket{S_t^*f, S_t^*g}}\left(-\frac d{dt}\braket{S_t^*f, S_t^*g}\right)\, dt = 1+e^{\braket{f,g}}-1 = \braket{e(f),e(g)}, \quad f,g\in L^2(\mathbb R_+).
    \end{align*}
\end{proof}
Hence the operator $W$ can be continued to an isometric operator on $H$.
\begin{proposition}
All possible values of the expression $ J(g):=
     \int_0^{+\infty}e(S_t^*g)\otimes \chi(dt) g$ span a dense subspace in $H\otimes L^2(\mathbb R_+)$.
\end{proposition}
\begin{proof}
    Fix an interval $[b,c)\subset \mathbb R_+$ and a smooth bounded function $f\in L^2(\mathbb R_+)$ such that $f(0)\neq 0$, $f'(0)\neq 0$, $t\in\mathbb R_+$. Take points $\{x_i\}_{i=0}^n$,  $b = x_0<x_1<\ldots<x_{n}=c$  dividing $[b,c]$ into $n$ equal parts. Then
    \begin{equation*}
        J(S_{x_k}f)-J(\chi[x_{k+1},+\infty)S_{x_k}f)=\int_{x_k}^{x_{k+1}}e(S_t^*S_{x_k}f)\otimes \chi(dt) S_{x_k}f.
    \end{equation*} 
    Notice that
    \begin{align*}
    &\left\|\int_{x_k}^{x_{k+1}}e(S_t^* S_{x_k}f)\otimes \chi(dt) S_{x_k}f- \int_{x_k}^{x_{k+1}}e(f)\otimes \chi(dt) S_{x_k}f \right\|\\
    &\le\max_{\tau\in[0,\frac{c-b}n]} \left\|e(S_{\tau}f) - e(f) \right  \|\cdot\max_{\tau\in[0,\frac{c-b}n]}|f(\tau)|\cdot\sqrt{\frac{c-b}n}\\
    &\le\max_{\tau\in[0,\frac{c-b}n]} \left\|e(S_{\tau}f) - e(f) \right  \|\cdot2|f(0)|\cdot\sqrt{\frac{c-b}n}
    \end{align*}
    provided that $n$ is large enough.
    Next,
    \begin{align*}
        \left\| \int_{x_k}^{x_{k+1}}e(f)\otimes \chi(dt) f - f(0)e(f)\otimes \chi[x_k,x_{k+1})\right\|&\le\|e(f)\|\cdot \max_{\tau\in[0,\frac{c-b}n]}|f(\tau)-f(0)|\cdot \sqrt{\frac{c-b}n}\\
        &\le \|e(f)\|\cdot 2|f'(0)|\cdot \left(\frac{c-b}n\right)^{3/2}
    \end{align*}
    provided that $n$ is large enough.
    
    Hence,
    \begin{align*}
       & \left\|\sum_{k=1}^{n-1}( J(S_{x_k}f)-J(\chi[x_{k+1},+\infty)S_{x_k}f)) - e(f)\otimes \chi[b,c)\right\|^2\le
        \\
       & \le n\cdot \left(
        \max_{\tau\in[0,\frac{c-b}n]} \left\|e(S_{\tau}f) - e(f) \right  \|\cdot2|f(0)|\cdot\sqrt{\frac{c-b}n}+\|e(f)\|\cdot 2|f'(0)|\cdot \left(\frac{c-b}n\right)^{3/2}
        \right)^2\to 0\text{ as } n\to \infty.
    \end{align*}
    
    Therefore $f\otimes \chi[b,c)$ belongs to the closure of  $\mathrm{Span}\left\{ J(g)\colon g\in L^2(\mathbb R_+)\right\}$ and thus this closure coincides with $H\otimes L^2(\mathbb R_+)$.
    \end{proof}
    
 Hence we obtain $W$ to be an isometric isomorphism. For an object (vector, operator or superoperator) $h\in H$ we denote $h^{\wedge}$ its isomorphic (with respect to $W$) object in $H^{\wedge}$. The inverse correspondence for an object $h'$ from $H^{\wedge}$ will be referred to as $h'^{\vee}$.

\begin{proposition}
    The semigroups $\mathcal T_{t}^\wedge$ and $\breve{\mathcal T}_t^\wedge$ isomorphic to the above ones act on $X\in B(H^\wedge)$ as
\begin{align*}
    \mathcal T_{t}^\wedge X &=  \Big(1\oplus(I\otimes S_t)\Big)X\Big(1\oplus(I\otimes S_t^*)\Big),\\
    \breve{\mathcal T}_{t}^\wedge X&=
    {\mathcal T}_{t}^\wedge X+0\oplus\left( \int_0^t\breve{\mathcal T}_{t-s}X^\vee\otimes \chi(ds)\right)  .
\end{align*}
\end{proposition}
\begin{proof}
First statement follows from the equalities
\begin{align*}
    \mathcal T_{t}^\wedge\ket{e(f)^\wedge}\bra{e(g)^\wedge}=&\ket{e(S_tf)^\wedge}\bra{e(S_tg)^\wedge}\\
    =&\Ket{\Omega\oplus\left(  \int_0^{+\infty}e(S_\tau^*S_tf)\otimes \chi(d\tau) S_tf\right)}\Bra{\Omega\oplus\left(  \int_0^{+\infty}e(S_\tau^*S_tg)\otimes \chi(d\tau) S_tg\right)}\\
    =&\Ket{\Omega\oplus\left(  \int_t^{+\infty}e(S_tS_\tau^*f)\otimes \chi(d\tau) S_tf\right)}\Bra{\Omega\oplus\left(  \int_t^{+\infty}e(S_tS_\tau^*g)\otimes \chi(d\tau) S_tg\right)}\\
    =&\Big(1\oplus(I\otimes S_t)\Big)\ket{e(f)^\wedge}\bra{e(g)^\wedge}\Big(1\oplus(I\otimes S_t^*)\Big).
\end{align*}
In order to prove the second statement, compute the matrix element of the right side with respect to exponential vectors:
\begin{align*}
    &\Braket{e(f)^\wedge,\left(
    {\mathcal T}_{t}^\wedge X+0\oplus\left( \int_0^t\breve{\mathcal T}_{t-s}X^\vee\otimes \chi(ds)\right)\right)e(g)^\wedge}= \Braket{e(f), \mathcal T_t X^\vee e(g)}\\
    &+\int_0^t\Braket{e(S_s^*f), \breve{\mathcal T}_{t-s}X^\vee e(S_s^*g)}\overline{f(s)}g(s)\, ds\\
    &= \Braket{
    e(S_t^*f),X^\vee e(S_t^*g)
    }+\int_0^t\exp{\left(\int_s^t\overline {f(\tau)}g(\tau)\, d\tau\right)}\Braket{
    e(S_t^*f),X^\vee e(S_t^*g)
    }\overline{f(s)}g(s)\, ds\\
    &=\Braket{
    e(S_t^*f),X^\vee e(S_t^*g)
    }\left(1+\exp{\left(\int_0^t\overline {f(\tau)}g(\tau)\, d\tau\right)}-1\right) = \Braket{e(f),\breve{\mathcal T}_tX^\vee e(g)}.
\end{align*}
\end{proof}

Now define the operator-valued measure $\mathfrak M$ on $B(H)$ by the formula
\begin{equation}\label{eq:Mdef}
    \mathfrak M\big([a,b)\big)Y =  \left[0\oplus\Big( Y\otimes \chi[a,b)\Big)\right]^\vee,\quad Y\in B(H).
    \end{equation}
The isomorphic measure in $H^\wedge$ equals
\begin{equation*}
     \mathfrak M^{\wedge}\big([a,b)\big)Y =  0\oplus\Big( Y^{\vee}\otimes \chi[a,b)\Big),\quad Y\in B(H^\wedge).
\end{equation*}
\begin{proposition}
The preadjoint measure $\mathfrak M_*$ acts on exponential rank one operators as
\begin{equation*}
    \mathfrak M_*\big([a,b)\big) \ket{e(f)}\bra{e(g)}=\int_a^b \ket{e(S_t^*f)}\bra{e(S_t^*g)}\overline{g(t)} f(t)\, dt.
\end{equation*}
\end{proposition}
\begin{proof} Take $Y\in B(H)$. Then
\begin{align*}
    \Braket{e(g), \mathfrak M\big([a,b)\big)Y e(f) } =&\Braket{e(g)^\wedge, \left[\mathfrak M\big([a,b)\big)Y\right]^{\wedge} e(f)^\wedge}\\
    =&\Braket{\int_0^{+\infty}e(S_t^*g)\otimes \chi(dt) g,\Big(
    Y
    \otimes \chi[a,b)\Big)
    \int_0^{+\infty}e(S_t^*f)\otimes \chi(dt) f} \\
    =&\Braket{\int_0^{+\infty}e(S_t^*g)\otimes \chi(dt) g,
    \int_a^bYe(S_t^*f)\otimes \chi(dt) f}\\
    =& \int_a^b\Braket{e(S_t^*g),Ye(S_t^*f)}\overline{g(t)}f(t)\, dt\\ 
    =& \Tr Y\int_a^b \ket{e(S_t^*f)}\bra{e(S_t^*g)}\overline{g(t)} f(t)\, dt.
\end{align*}
\end{proof}

\begin{proposition}
$\mathfrak M $ satisfies the covariant property with respect to $\mathcal T_t$.
\end{proposition}
\begin{proof}
Indeed,
\begin{align*}
    \mathcal T_t^\wedge \mathfrak M^\wedge \big([a,b)\big)Y = 0\oplus (Y^\vee\otimes S_t\chi[a,b)S_t^* ) =0\oplus (Y^\vee\otimes\chi[a+t,b+t)) \\
    =\mathfrak M^\wedge \big([a+t,b+t)\big)Y,\quad Y\in B(H^\wedge).
\end{align*}
\end{proof}
\begin{proposition}
    The measure $\mathfrak M$ and the semigroup $\breve{\mathcal T}_t$ satisfy the equation~\eqref{equ} up to isomorphism $W$.
\end{proposition}
\begin{proof}
Indeed,
\begin{align*}
    \breve{\mathcal T}_t^\wedge X - \int_0^t\mathfrak M^\wedge(ds)\breve{\mathcal T}_{t-s}^\wedge X = 
    {\mathcal T}_{t}^\wedge X+0\oplus\left( \int_0^t\breve{\mathcal T}_{t-s}X^\vee\otimes \chi(ds)\right)- 0\oplus \left(\int_0^t \breve{\mathcal T}_{t-s}X^\vee\otimes \chi(ds)\right) \\
    = \mathcal T^\wedge X, \quad X\in B(H^\wedge) .
\end{align*}
\end{proof}


\section{Conclusion}\label{Sec:Conclusions}
In this work, we have studied the problem of restoring states of quantum systems from the dynamics induced by quantum dynamical semigroups perturbed by covariant measures. We introduce definitions of an (unbounded) covariant operator-valued measure on the half-axis and divide the set of all such measures into two subclasses. The first subclass consists of measures having (at least unbounded) density with respect to the Lebesgue measure on the half-axis, and the second subclass consists of measures singular with respect to the Lebesgue measure. Examples of measures belonging to both subclasses are given. Next we consider perturbations of no-event quantum dynamical semigroups by means of measures. It is shown that the Lindbladian form of the generator fits into the concept of perturbation, but the class under consideration contains a broader class of semigroups. In the last section of the paper, an operator-valued measure is constructed in the symmetric Fock space, covariant with respect to the semiflow of shifts. Various physically meaningful examples are provided, including that in the single particle case corresponds to the Koopman-von Neumann dynamics of a classical oscillator with bounded phase space~\cite{McCaulPRER2019} and the model of the perturbed dynamics in the symmetric (boson) Fock space. This model, which as an example describes quantum optics processes like photon transmittance along fibers, is explicitly constructed as covariant measure for a semiflow of shifts and its perturbation in the symmetric Fock space, and its properties are investigated. 

\section*{Acknowledgment} This work was funded by the Ministry of Science and Higher Education of the Russian Federation (grant number 075-15-2020-788) and performed at the Steklov Mathematical Institute of the Russian Academy of Sciences.

\section*{Availability of data and materials} Data sharing is not applicable to this article as no datasets were generated or analyzed during the current study.

\section*{Conflict of interests}

The author declares no conflicts of interest.

\end{document}